\documentclass[article,11pt]{article}
\usepackage{amsmath,amssymb,amsthm,graphicx,setspace}
\usepackage{amsfonts}
\date{}

\newtheorem{theorem}{Theorem}[section]

\theoremstyle{definition}

\theoremstyle{remark}
\newtheorem{remark}[theorem]{Remark}
\numberwithin{equation}{section}
\begin{document}
\title{ Numerical analysis for Spread option pricing model in illiquid underlying asset market: full feedback model}
\author{ \small A. R. Yazdanian\thanks{Corresponding author. Tel: +98 912 848 7540;
E-mail: ahmadreza.yazdanian@gmail.com} $^{a}$, \small Traian A. Pirvu $^b$  }
\maketitle\centerline{\it\small  $^a$ Faculty of Mathematics, Statistics and Computer Science ,
Semnan University, Tehran, Iran.}
\maketitle\centerline{\it\small  $^b$ Department of Mathematics and Statistics, McMaster University, Hamilton, Canada.}

\begin{abstract}
 
 This paper performs the numerical analysis and the computation
 of a  Spread option in a market with imperfect liquidity. The number of shares traded in the stock market has a direct impact on the stock's price. Thus, we consider a \textit{full-feedback model} in which price impact is fully incorporated into the model. The price of a Spread
 option is characterize by a nonlinear partial differential equation. This is reduced to
 linear equations by asymptotic expansions. The Peaceman-Rachford scheme as an alternating direction implicit method is employed to solve the linear equations numerically.  We discuss the stability and the convergence of the numerical scheme. Illustrative examples are included to demonstrate the validity and applicability of the presented method. Finally we provide a numerical analysis of the illiquidity effect in replicating an European Spread option; compared to the Black-Scholes case, a trader generally buys more stock to replicate this option.
\end{abstract}


\section{Introduction}\label{section1}
Classical asset pricing theory assumes that traders act as price takers, that is,
they have no impact on the prices paid
or received. The relaxation of this assumption and its impact
on realized returns in asset pricing models is called liquidity risk. Consistent with this discussion, most of the option pricing models assume that an option trader cannot affect the price in trading the underlying asset to replicate the option payoff, regardless of her trading size. The papers of Black and Scholes \cite{Black}, and much of the work undertaken in mathematical finance has been done under this underlying assumption (which is reasonable only in a perfectly liquid market). \\
 In presence of a price impact, the replication of an option becomes more involved.
 The first issues is whether or not the option is perfectly replicable. Second, one
 has to find out how the presence of price impact affects the the replicating costs. 
 This encouraged researchers to develop a Black-Scholes model with price impact due to a large trader who is able to move the price by his/her actions. 
 An excellent survey of these research can be found in \cite{Frey, Wilmott, Yong, Kristoffer}.
 Most of these research works discussed how the price impact affects the replication of an option written on a single stock.
 
 The purpose of our paper is to investigate the effects of imperfect liquidity on the replication of an European  Spread option by a typical option trader in a full-feedback model
 (any trade will impact the prices of the underlyings).
 Spread option is a simple example of multi-assets derivative, whose payoff is the difference between the prices of two or more assets; for instance let the prices of two underlying assets at time $t\in[0,T]$ be $S_1(t)$ and $S_2(t)$, then the payoff function of an European Spread
option with maturity $T$ is $[S_1(T)-S_2(T)-k]^+$ (here $k$ is the strike of the option and the function $x^+$ is defined as $ x^+ = max(x, 0)$). Therefore the holder of an European Spread option has the right but not the obligation to buy the spread $S_1(T)-S_2(T)$ at the prespecified price $k$ and maturity $T$. In general, there is no any analytical formula for the price of multi-assets options (even in models with perfect liquidity). The only exception is Margrabe formula for exchange options (Spread options with a strike of zero)\cite{Margrabe}. Margrabe derived a Black and Scholes type solution for this class of options as follows
\begin{equation}\label{1-1}
C^M = S_1\Phi(d_+) - S_2\Phi(d_-),
\end{equation}
where $\Phi$ is the standard normal cumulative distribution function and
\begin{equation}\label{1-2}
\begin{split}
&d_{\pm}=\dfrac{\ln(S_1/S_2)}{\sigma\sqrt{T}}\pm \dfrac{1}{2}\sigma\sqrt{T},\\
&\sigma=\sqrt{\sigma_1^2+\sigma^2_2-2\rho\sigma_1\sigma_2}.
\end{split}
\end{equation}
 Since that a linear combination of correlated lognormals is not lognormal, for non-zero strikes, there is no closed form Spread option valuation formula under the multivariate lognormal model. People rely on approximation formulas and
numerical methods for Spreads valuation. Kirk\cite{Kirk} suggested the following analytical approximation for a Spread option with payoff $(S_1(T) - S_2(T) -k)^+$
\begin{equation}\label{1-3}
C^M = S_1\Phi(d_+) - (S_2+k)\Phi(d_-),
\end{equation}
where 
\begin{equation}\label{1-4}
\begin{split}
&d_{\pm}=\dfrac{\ln(S_1/(S_2+k))}{\sigma\sqrt{T}}\pm \dfrac{1}{2}\sigma\sqrt{T},\\
&\sigma=\sqrt{\sigma_1^2+(\dfrac{S_2}{S_2+k})^2\sigma^2_2-2\dfrac{S_2}{S_2+k}\rho\sigma_1\sigma_2}.
\end{split}
\end{equation}
This formula provides a good approximation of Spread option prices when the
strike $k$ is not far from zero. All these works on spread options assume a
model with perfect liquidity.

Several Spread options are traded in the markets. Some popular Spread option products are : fixed Income Spread options and  commodity Spread options (including Crush Spread options, Crack Spread options, Spark Spread options).
In this work, we will focus our interest on Oil Markets and more specifically one of the most frequently quoted Spread options which are Crack Spreads. A Crack Spread
represents the differential between the price of crude oil and petroleum products (gasoline or heating oil). The underlying indexes comprise
futures prices of crude oil, heating oil and unleaded gasoline. Details of crack Spread
options can be found in the New York Mercantile Exchange (NYMEX) Crack Spread
Handbook \cite{New}. In the oil markets with finite liquidity, trading does affect the underlying assets price. In our study, we are going to investigate the effects of impact price on Spread option pricing in oil markets, when trading affects only the crude oil price and not petroleum products .

This paper is organized as follows: Section\ref{section2} discusses the general framework. In Section\ref{section3}  we apply an asymptotic expansion for the full nonlinear 
partial differential equation which characterizes the spread option price. 
In Section\ref{section4}, we propose a numerical method for the linearized equation. 
( the Peaceman and Rachford  numerical scheme \cite{Peaceman} is employed). We  discuss the stability and convergence of this scheme. In Section\ref{section5}, we carry out several numerical experiments and provide a numerical analysis of the model for European Spread calls. Section\ref{section6} contains the concluding remarks. 

\section{Statement of the problem}\label{section2}
In this section we describe the setup used for pricing Spread options. Our model of a financial market, based on a filtered probability space $(\Omega,\mathcal{F},\lbrace\mathcal{F}_t\rbrace_{t \in [0,T]},\mathbf{P} )$ that satisfies the usual conditions,
consists of two assets. Their prices are modeled by a
two--dimensional diffusion process $S(t)=(S_1(t),S_2(t))$. All the stochastic processes in this work are assumed to be $\lbrace\mathcal{F}_t\rbrace_{t\geq 0}$-adapted. Their dynamics are given by the following stochastic differential equations, in which $W(t)=(w_1(t),w_2(t))$ is defined a two--dimensional standard correlated Brownian motion with correlation $\rho$, and $\lbrace\mathcal{F}_t\rbrace_{t \in [0,T]}$ being its natural filtration augment by all $\mathbf{P}$-null sets:
\begin{equation}\label{2-0}
\dfrac{dS_i(t)}{S_i(t)}=\mu_i(t,S_i(t)) dt+\sigma_i(t,S_i(t))dw_i(t);~ i=1,2.
\end{equation} 
Here $\mu_i(t,S_i(t))$ and $\sigma_i(t,S_i(t))$ are the expected return and the volatility of stock $i$ in the absence of price impact. It is possible to add a forcing term, $\lambda(t,S_1(t)),$i.e.
\begin{flalign}\label{2-1}
\begin{split}
dS_1(t)&=\mu_1(t,S_1(t))S_1(t) dt+\sigma_1(t,S_1(t))S_(t)dw_1(t)+\lambda(t,S_1)d\Delta_1(t),\\
dS_2(t)&=\mu_2(t,S_2(t)) S_2(t)dt+\sigma_2(t,S_2(t))S_2(t)dw_2(t),
\end{split}
\end{flalign}
($\lambda (t,S_1(t))$ is the price impact function on the first stock). The term $\lambda(t,S_1(t))d\Delta_1(t)$ represents the price impact
of the trading strategy $\Delta_1(t).$ We note that the classical Black-Scholes model is a special case of
this model with $\lambda(t,S_1(t))=0$.
Our aim is to price a Spread option in this illiquid market model. The option's payoff at maturity $T$ (a call at this case) is:
\begin{flalign}\label{2-2}
h(S_1(T),S_2(T))=(S_1(T)-S_2(T)-k)^+,
\end{flalign}
where $k$ is the strike price. The well-known generalized Black-Scholes equation (more details \cite{Duffie}) is used to characterize the spread option price within
our full-feedback model for $S_1.$ This leads to the nonlinear PDE 
  \begin{equation}\label{2-3}
\begin{split}
&\frac{\partial V}{\partial t}+
\dfrac{1}{2(1-\lambda\dfrac{\partial^2 V}{\partial S_1^2})^2}(\sigma_1^2S_1^2+\lambda^2\sigma_2^2S_2^2(\dfrac{\partial^2 V}{\partial S_1\partial S_2})^2
+2\rho \sigma_1\sigma_2S_1S_2 \lambda \dfrac{\partial^2 V}{\partial S_1\partial S_2})\frac{\partial^2 V}{\partial S_1^2}+\dfrac{1}{2}\sigma_{2}^2S_{2}^{2}\frac{\partial^2 V}{\partial S_2^2}+\dfrac{1}{1-\lambda\dfrac{\partial^2 V}{\partial S_1^2}}\\
&\times (\sigma_1\sigma_2\rho S_1S_2+\lambda \sigma_2^2S_2^2\dfrac{\partial^2 V}{\partial S_1\partial S_2})\frac{\partial^2 V}{\partial S_1\partial S_2}
+r(S_{1}\frac{\partial V}{\partial
S_{1}}+S_{2}\frac{\partial V}{\partial S_{2}})-rV=0,~~~0<S_{1},S_{2}<\infty, 0\leq t<T,\\
&V(T,S_{1},S_{2})=h(S_{1},S_{2}),~~~ 0<S_{1},S_{2}<\infty, 
\end{split}
\end{equation}
with the terminal condition  in $T$. Notice that the classical Black-Scholes model for Spread option is a special case of this model with $\lambda=0$.

\section{Matched Asymptotic Expansions}\label{section3}
 In this section we use a matched asymptotic expansion technique to linearize (\ref{2-3}). 
 For this purpose we let $\lambda(t, S_1(t))=\varepsilon \widehat{\lambda}(t, S_1(t)),$ so that (\ref{2-3}) becomes 
\begin{flalign}\label{5-0}
\begin{split}
&\frac{\partial V}{\partial t}+
\dfrac{1}{2(1-\varepsilon \widehat{\lambda}\dfrac{\partial^2 V}{\partial S_1^2})^2}(\sigma_1^2S_1^2+(\varepsilon \widehat{\lambda})^2\sigma_2^2S_2^2(\dfrac{\partial^2 V}{\partial S_1\partial S_2})^2
+2\rho \sigma_1\sigma_2S_1S_2 \varepsilon \widehat{\lambda} \dfrac{\partial^2 V}{\partial S_1\partial S_2})\frac{\partial^2 V}{\partial S_1^2}\\
&+\dfrac{1}{2}\sigma_{2}^2S_{2}^{2}\frac{\partial^2 V}{\partial S_2^2}+\dfrac{1}{1-\varepsilon \widehat{\lambda}\dfrac{\partial^2 V}{\partial S_1^2}}(\sigma_1\sigma_2\rho S_1S_2+\varepsilon \widehat{\lambda} \sigma_2^2S_2^2\dfrac{\partial^2 V}{\partial S_1\partial S_2})\frac{\partial^2 V}{\partial S_1\partial S_2}
+r(S_{1}\frac{\partial V}{\partial
S_{1}}+S_{2}\frac{\partial V}{\partial S_{2}})-rV=0.
\end{split}
\end{flalign}
 
By replacing $V(t,S_{1},S_{2})$ in the equation (\ref{5-0}) with
\begin{flalign}\label{5-2}
V(t,S_{1},S_{2})\sim V^0(t,S_{1},S_{2})+\varepsilon V^1(t,S_{1},S_{2})+o(\varepsilon^2),
\end{flalign}
we get
\begin{flalign}\label{5-3}
\begin{split}
&\frac{\partial (V^0+\varepsilon V^1)}{\partial t}+
\dfrac{\frac{\partial^2  (V^0+\varepsilon V^1)}{\partial S_1^2}}{2(1-\varepsilon \widehat{\lambda}\dfrac{\partial^2  (V^0+\varepsilon V^1)}{\partial S_1^2})^2}(\sigma_1^2S_1^2+(\varepsilon \widehat{\lambda})^2\sigma_2^2S_2^2(\dfrac{\partial^2  (V^0+\varepsilon V^1)}{\partial S_1\partial S_2})^2\\
&+2\rho \sigma_1\sigma_2S_1S_2 \varepsilon \widehat{\lambda} \dfrac{\partial^2  (V^0+\varepsilon V^1)}{\partial S_1\partial S_2})
+\dfrac{1}{1-\varepsilon \widehat{\lambda}\dfrac{\partial^2  (V^0+\varepsilon V^1)}{\partial S_1^2}} (\sigma_1\sigma_2\rho S_1S_2
+\varepsilon \widehat{\lambda} \sigma_2^2S_2^2\dfrac{\partial^2  (V^0+\varepsilon V^1)}{\partial S_1\partial S_2})\frac{\partial^2  (V^0+\varepsilon V^1)}{\partial S_1\partial S_2}\\
&+\dfrac{1}{2}\sigma_{2}^2S_{2}^{2}\frac{\partial^2  (V^0+\varepsilon V^1)}{\partial S_2^2}+r(S_{1}\frac{\partial  (V^0+\varepsilon V^1)}{\partial
S_{1}}+S_{2}\frac{\partial  (V^0+\varepsilon V^1)}{\partial S_{2}})-r (V^0+\varepsilon V^1)+o(\varepsilon^2)=0.
\end{split}
\phantom{\hspace{9cm}}
\end{flalign}
Using Maclaurin series expansion on some terms of above equation, we have:
\begin{flalign}\label{5-4}
\begin{split}
\dfrac{1}{2(1-\varepsilon \widehat{\lambda}\dfrac{\partial^2  (V^0+\varepsilon V^1)}{\partial S_1^2})^2}&=
\dfrac{1}{2}+\varepsilon \widehat{\lambda} \dfrac{\partial^2 V^0}{\partial S_1^2 }+o(\varepsilon^2)\\
\dfrac{1}{1-\varepsilon \widehat{\lambda}\dfrac{\partial^2  (V^0+\varepsilon V^1)}{\partial S_1^2}}&=
1+\varepsilon \widehat{\lambda}\dfrac{\partial^2 V^0}{\partial S_1^2 }+o(\varepsilon^2).
\end{split}
\end{flalign}
Therefore we get the following linear equation for $V^0(t,S_1,S_2)$ 
\begin{flalign}\label{5-5}
\begin{split}
&\frac{\partial V^0}{\partial t}+\frac{
\sigma_{1}^{2}S_{1}^{2}}{2}\frac{\partial^2 V^0}{\partial S_1^2}+ \frac{
\sigma_{2}^{2}S_{2}^{2}}{2}\frac{\partial^2 V^0}{\partial S_2^2}+
 \sigma_{1}\sigma_{2}S_{1}S_{2}\rho \frac{\partial^2 V^0}{\partial S_1\partial S_2}+r[S_{1}\frac{\partial V^0}{\partial S_1}+S_{2}\frac{\partial V^0}{\partial S_2}]-rV^0= 0,\\
&V^0(T,S_{1},S_{2})=h(S_{1},S_{2}),~~~~~ 0<S_{1},S_{2}<\infty,
\end{split}
\end{flalign}
and for $V^1(t,S_1,S_2)$ 
\begin{equation}\label{5-6}
\begin{split}
&\frac{\partial V^1}{\partial t}+\frac{
\sigma_{1}^{2}S_{1}^{2}}{2}\frac{\partial^2 V^1}{\partial S_1^2}+ \frac{
\sigma_{2}^{2}S_{2}^{2}}{2}\frac{\partial^2 V^1}{\partial S_2^2}+
 \sigma_{1}\sigma_{2}S_{1}S_{2}\rho \frac{\partial^2 V^1}{\partial S_1\partial S_2}+r[S_{1}\frac{\partial V^1}{\partial S_1}+S_{2}\frac{\partial V^1}{\partial S_2}]-rV^1= G,\\
&V^1(T,S_{1},S_{2})=0,~~~~~ 0<S_{1},S_{2}<\infty.
\end{split}
\end{equation}
Here 
\begin{equation}\label{5-7}
\begin{split}
G&=-\widehat{\lambda}(2\rho \sigma_1\sigma_2S_1S_2 \dfrac{\partial^2 V^0}{\partial S_1\partial S_2}\dfrac{\partial^2 V^0}{\partial S_1^2}+\sigma_1^2S_1^2(\dfrac{\partial^2 V^0}{\partial S_1^2})^2 +\sigma_2^2S_2^2(\dfrac{\partial^2 V^0}{\partial S_1\partial S_2})^2).
\end{split}
\end{equation}
 In the following, we apply a numerical scheme for comupting $V^0(t, S_1, S_2)$ and $V^1(t,S_1,S_2)$.

\section{Numerical Solution of the Partial Differential Equations}\label{section4}
 \subsection{The Alternating Direction Implicit}
In this section, we present a numerical method for solving the partial differential equations:
\begin{flalign}\label{6-0}
\begin{split}
&\frac{\partial V^0}{\partial t}+\frac{
\sigma_{1}^{2}x^{2}}{2}\frac{\partial^2 V^0}{\partial x^2}+ \frac{
\sigma_{2}^{2}y^{2}}{2}\frac{\partial^2 V^0}{\partial y^2}+
 \sigma_{1}\sigma_{2}xy\rho \frac{\partial^2 V^0}{\partial x\partial y}
+r[S_{1}\frac{\partial V^0}{\partial x}+S_{2}\frac{\partial V^0}{\partial y}]-rV^0= 0,\\
&V^0(T,x,y)=h(x,y),~~~~~ 0<x,y<\infty,
\end{split}
\end{flalign}
and
\begin{equation}\label{6-1}
\begin{split}
&\frac{\partial V^1}{\partial t}+\frac{
\sigma_{1}^{2}x^{2}}{2}\frac{\partial ^2 V^1}{\partial x^2}+\frac{
\sigma_{2}^{2}y^{2}}{2}\frac{\partial^2 V^1}{\partial y^2} +\rho \sigma_{1}\sigma_{2}xy\frac{\partial^2 V^1}{\partial x\partial y}+r[S_{1}\frac{\partial V^1}{\partial x}+y\frac{\partial V^1}{\partial y}]-rV^1=G,\\
 &V^1(T,x,y)=0,~~~ 0<x,y<\infty .
\end{split}
\end{equation}
Functions $V^0(t,x,y)$ and $ V^1(t,x,y)$ are defined on $[0,T]\times [0,\infty)\times [0,\infty)$. 
To simplify notations we write:
\begin{equation}\label{6-2}
L=\dfrac{\partial}{\partial t}+A_{x}+A_{y}+A_{xy},
\end{equation}
where 
\begin{equation}\label{6-3}
\begin{split}
&A_{x}=\dfrac{1}{2}\sigma_{1}^2x^2\dfrac{\partial^2}{\partial x^2}+rx\dfrac{\partial}{\partial x}-r\Theta , \\
&A_{y}=\dfrac{1}{2}\sigma_{2}^2y^2\dfrac{\partial^2}{\partial y^2}+ry\dfrac{\partial}{\partial y}-r(1-\Theta) ,\\
&A_{xy}=\sigma_{1}\sigma_{1}xy \rho \dfrac{\partial^2}{\partial x\partial y},
\end{split}
\end{equation}
and $0\leq \Theta \leq 1$. While symmetry considerations might speak for an $\Theta=\dfrac{1}{2}$
, it is computationally simpler to use  $\Theta=0$ or $\Theta=1$, i.e., nclude the $rV-$ term fully in one of the two operators.
Hence, we can write
\begin{equation}\label{6-4}
\left\{
	\begin{array}{ll}
		 LV^0(t,x,y)=0, &\\
		LV^1(t,x,y)=G, &
	\end{array}
\right.
 \end{equation}
where 
\begin{equation}\label{6-5}
\begin{split}
G&=-\widehat{\lambda}(2\rho \sigma_1\sigma_2xy \dfrac{\partial^2 f^0}{\partial x\partial y}\dfrac{\partial^2 f^0}{\partial x^2}+\sigma_1^2x^2(\dfrac{\partial^2 f^0}{\partial x^2})^2 +\sigma_2^2y^2(\dfrac{\partial^2 f^0}{\partial x\partial y})^2).
\end{split}
\end{equation}
In order to define a numerical solution for these equations, we need to truncate the spatial domain to a bounded area: $\{(x,y);0\leq x\leq x_{ max} , 0 \leq y \leq y_{max}\}$. Let us introduce a grid of points in the time interval and in the truncated spatial domain:
\begin{equation}\label{6-6}
\begin{split}
&t_l=l\Delta t,~~l=0,1,...L, ~~\Delta t=\dfrac{T}{L},\\
&x_m=m\Delta x,~~m=0,1,...M,~~\Delta x=\dfrac{x_{max}}{M},\\
&y_n=n\Delta y,~~n=0,1,...N,~~\Delta y=\dfrac{y_{max}}{N}.
\end{split}
\end{equation}
For the simplicity assume that $x_{max}=y_{max}$ and $\Delta x=\Delta y$. Functions  $ V^0(t,x,y)$ and $ V^1(t,x,y)$ evaluated at a point on the grid
are denoted as $V^{0,l}_{mn}=V^0(t_l,x_m,y_n)$ and $V^{1,l}_{mn}=V^1(t_l,x_m,y_n)$. If we need to refer to the solution at a specific time point, we will use notation  $V^{0,l}=V^0(t_l,x_m,y_n)$ and $V^{1,l}=V^1(t_l,x_m,y_n)$. Furthermore, let symbols $ A_{dx}, A_{dy}$ and $ A_{dxdy}$ denote second-order approximations
to the operators $A_{x}, A_{y}$ and $A_{xy}$. Since the differential operator can be split as in (\ref{6-3}) we can use Alternating Direction Implicit (ADI). The general idea is to split a time step into two and consider one operator or one space coordinate at a time. We implement the Peaceman-Rachford scheme. Let us begin by discretizing (\ref{6-0}) in the time-direction:
\begin{equation}\label{6-7}
\begin{split}
&V^0_t((l+1/2)\Delta t,x,y)=\frac{V^{0,l+1}-V^{0,l}}{\Delta t}+O(\Delta t ^2)\\
&(A_x+A_y+A_{xy})V^0=\dfrac{1}{2}A_x(V^{0,l+1}+V^{0,l})+\dfrac{1}{2}A_y(V^{0,l+1}+V^{0,l})+\dfrac{1}{2}A_{xy}(V^{0,l+1}+V^{0,l})+O(\Delta t ^2).
\end{split}
\end{equation}
Next Insert into (\ref{6-2}), multiply by $\Delta t$, and rearrange to obtain:
\begin{equation}\label{6-8}
\begin{split}
(I-\dfrac{1}{2}\Delta t A_x-\dfrac{1}{2}\Delta t A_y)V^{0,l}=(I+\dfrac{1}{2}\Delta t A_x+\dfrac{1}{2}\Delta t A_y)V^{0,l+1}+\dfrac{1}{2}\Delta t A_{xy}(V^{0,l+1}+V^{0,l})+O(\Delta t ^3),
\end{split}
\end{equation}
where $I$ denotes the identity operator. If we add $\dfrac{1}{4}\Delta t^2 A_xA_y V^{0,l}$  on the left side and $\dfrac{1}{4}\Delta t^2 A_xA_y V^{0,l+1}$  on the right side then we
commit an error which is $O(\Delta t^3)$ and therefore:
\begin{equation}\label{6-9}
\begin{split}
&(I-\frac{1}{2}\Delta t A_x)(I-\frac{1}{2}\Delta t A_y) V^{0,l}\\
&=(I+\frac{1}{2}\Delta t A_x)(I+\frac{1}{2}\Delta t A_y) V^{0,l+1}+\frac{1}{2}\Delta t A_{xy}( V^{0,l+1}+ V^{0,l})+O(\Delta t^3).
\end{split}
\end{equation}
We now discretize in the space coordinates replacing $A_x$ by $A_{dx}$, $A_y$ by $A_{dy}$ and $A_{xy}$ by $A_{dxdy}$
\begin{equation}\label{6-10}
\begin{split}
&(I-\frac{1}{2}\Delta t A_{dx})(I-\frac{1}{2}\Delta t A_{dy}) V^{0,l}\\
&=(I+\frac{1}{2}\Delta t A_{dx})(I+\frac{1}{2}\Delta t A_{dy}) V^{0,l+1}+\frac{1}{2}\Delta t A_{dxdy}( V^{0,l+1}+ V^{0,l})+O(\Delta t^3)+O(\Delta t \Delta x^2).
\end{split}
\end{equation}
This leads to the Peaceman-Rachford method \cite{Strikwerda}
\begin{equation}\label{6-11}
\begin{split}
&(I-\frac{\Delta t}{2}A_{d x})V^{0,l+1/2}=(I+\frac{\Delta t}{2}A_{dy})V^{0,l+1}+\alpha ,\\
&(I-\frac{\Delta t}{2}A_{dy})V^{0,l}=(I+\frac{\Delta t}{2}A_{dx})V^{0,l+1/2}+\beta ,
\end{split}
\end{equation}
where auxiliary function $V^{0,l+1/2}$  links above equations. We have introduced the values $\alpha$ and $\beta$ to take into account the mix derivative
term because it is not obvious how this term should be split. To align (\ref{6-11}) with (\ref{6-10}), we require that

\begin{equation}\label{6-12}
\begin{split}
(I+\frac{\Delta t}{2}A_{d x})\alpha +(I-\frac{\Delta t}{2}A_{d x})\beta =\frac{1}{2}\Delta t A_{dxdy}( V^{0,l+1}+ V^{0,l}),
\end{split}
\end{equation}
where a discrepancy of order $O(\Delta t^3)$ may be allowed with reference to a similar
term in (\ref{6-9}). One of the possible choices for  $\alpha$ and $\beta$ is 
\begin{equation}\label{6-13}
\alpha=\frac{\Delta t}{2}A_{dxdy}V^{0,l+1} ,~~~~\beta =\frac{\Delta t}{2}A_{dxdy} V^{0,l+1/2}.
\end{equation}
Finally, the Peaceman-Rachford scheme for $V^0$  in (\ref{6-0}) is obtain as follows
\begin{equation}\label{6-14}
\begin{split}
&(I-\frac{\Delta t}{2}A_{dx})V^{0,l+1/2}=(I+\frac{\Delta t}{2}A_{dy})V^{0,l+1}+\frac{\Delta t}{2}A_{dxdy}V^{0,l+1},\\
&(I-\frac{\Delta t}{2}A_{dy})V^{0,l}=(I+\frac{\Delta t}{2}A_{dx})V^{0,l+1/2}+\frac{\Delta t}{2}A_{dxdy}V^{0,l+1/2}.
\end{split}
\end{equation}
In a first step we calculate $V^{0,l+1/2}$ using $V^{0,l+1}$. This step is implicit in direction $x$. In a second step, defined by equations (\ref{6-14}), we use $V^{0,l+1/2}$  to calculate $V^{0,l}$. This step is implicit in the
direction of $y$. The Peaceman-Rachford scheme for $V^1$  in (\ref{6-1}) is obtained as follows: 
\begin{equation}\label{6-15}
\begin{split}
&(I-\frac{\Delta t}{2}A_{dx})V^{1,l+1/2}=(I+\frac{\Delta t}{2}A_{dy})V^{1,l+1}+\alpha , \\
&(I-\frac{\Delta t}{2}A_{dy})V^{1,l}=(I+\frac{\Delta t}{2}A_{dx})V^{1,l+1/2}+\beta ,
\end{split}
\end{equation}
where auxiliary function $V^{1,l+1/2}$  links above equations. To align (\ref{6-15}) with (\ref{6-10}) we require that
\begin{equation}\label{6-16}
\begin{split}
(I+\frac{\Delta t}{2}A_{d x})\alpha +(I-\frac{\Delta t}{2}A_{d x})\beta =\frac{1}{2}\Delta t A_{dxdy}( V^{1,l+1}+ V^{1,l})-\frac{1}{2}\Delta t(G^{l+1}+G^l).
\end{split}
\end{equation}
One of the possible choices for $\alpha$ and $\beta$ is 
\begin{equation}\label{6-17}
\alpha=\frac{\Delta t}{2}A_{dxdy}V^{1,l+1}-\frac{\Delta t}{2}G^{l+1} ,~~~~\beta =\frac{\Delta t}{2}A_{dxdy} V^{0,l+1/2}-\frac{\Delta t}{2}G^{l}.
\end{equation}
The Peaceman-Rachford scheme for $V^1$  of (\ref{6-1}) is obtained as follows:
\begin{equation}\label{6-18}
\begin{split}
&(I-\frac{\Delta t}{2}A_{dx})V^{1,l+1/2}=(I+\frac{\Delta t}{2}A_{dy})V^{1,l+1}+\frac{\Delta t}{2}A_{dxdy}V^{1,l+1}-\frac{\Delta t}{2}G^{l+1},\\
&(I-\frac{\Delta t}{2}A_{dy})V^{1,l}=(I+\frac{\Delta t}{2}A_{dx})V^{1,l+1/2}+\frac{\Delta t}{2}A_{dxdy} V^{0,l+1/2}-\frac{\Delta t}{2}G^{l}.
\end{split}
\end{equation}
In a first step calculate $V^{1,l+1/2}$ using $V^{1,l+1}$. This step is implicit in direction $x$. In ta second step, defined by equations (\ref{6-18}), we use $V^{1,l+1/2}$  to calculate $V^{1,l}$ . This step is implicit in the
direction of $y$. \\
Notice that due to the use of centered approximations of the derivatives, at $x_{0}=y_0=0$, $x_m=x_{max}$ and $y_n=y_{max}$, there appear external fictitious nodes $x_{-1}=-\Delta x$, $y_{-1}=-\Delta y$, $x_{M+1}=(M+1)\Delta x$ and $y_{N+1}=(N+1)\Delta$. The approximations in these nodes are obtained by using linear interpolation. Thus we have the following relations 
\begin{equation}\label{6-19}
\begin{split}
V^{0,l}_{-1,n}&=2V^{0,l}_{0,n}-V^{0,l}_{1,n},~~~ V^{0,l}_{M+1,n}=2V^{0,l}_{M,n}-V^{0,l}_{M-1,n};~~ n=1(1)N, \\
V^{0,l}_{m,-1}&=2V^{0,l}_{m,0}-V^{0,l}_{m,1}, ~~V^{0,l}_{m,N+1}=2V^{0,l}_{m,N}-V^{0,l}_{m,N-1};~~ m=1(1)M.
\end{split}
\end{equation}
Similarly, we can write the same relations in terms of $V^{1,l}$. Now all values $V^{0,l}_{m,n}$ and $V^{1,l}_{m,n}$ are available. By repeating this procedure for $l=L-1, L-2,...,0$ we obtain $V^{0}_{m,n}$ and $V^{1}_{m,n} $at all time points. The price of a Spread option at time $t_0=0$ can be approximated as :
\begin{equation}\label{6-20}
V(t_0, x, y)\approx V^{0}(t_0,x, y)+\varepsilon V^{1}(t_0,x, y).
\end{equation}

\subsection{Stability and Convergence of the Numerical Scheme}
In this section, we discuss stability and convergence of the numerical schemes introduced in Section 4.1. First we analyze the stability of the Peaceman-Rachford. In this case, we can use the Von Neumann analysis to establish the conditions for stability. This approach was described in \cite{Strikwerda}(Chapter 2.2). The Von Neumann analysis is based on calculating the amplification factor of a scheme ($g$), and deriving conditions under which $|g| \leq 1$.
\begin{theorem}
A one-step finite difference scheme (with constant coefficients) is stable in
a stability region $\Lambda$(any bounded nonempty
region of the first octant of $R^3$ that has the origin as an accumulation point) if and only if there is a constant $c$ (independent of $\theta$, $\phi$, $dt$, $dx$ and $dy$)
such that
\begin{equation}\label{8-0}
g|(\theta,\phi, dt,dx, dy)|\leq 1+cdt.
\end{equation}
Here $g(\theta,\phi, dt,dx, dy)$ is amplification factor of scheme with $(dt, dx, dy) \in \Lambda $. If $ g(\theta,\phi, dt,dx, dy)$ is independent of $dx$, $dy$ and $dt$, the above stability condition can be replaced with the restricted stability condition
\begin{equation}\label{8-1}
|g(\theta, \phi)|\leq 1.
\end{equation}
\end{theorem}
\begin{proof}
See \cite{Strikwerda}.
\end{proof}

\begin{remark}
This theorem shows that to determine the stability of a finite difference scheme with the constant coefficient, we only need to consider the amplification factor $g$. This theorem does not apply directly to problems with variable coefficients. Nonetheless, the stability conditions obtained for constant coefficient
schemes can be used to give stability conditions for the same scheme applied to
equations with variable coefficients. The general procedure is that one considers each of the frozen coefficient problems
arising from the scheme. The frozen coefficient problems are the constant coefficient
problems obtained by fixing the coefficients at their values attained at each point in the
domain of the computation. If each frozen coefficient problem is stable, then the variable
coefficient problem is also stable. The interest reader can see the proof of this result in \cite{Lax,Wade}. 
\end{remark}
For finding the amplification factor, a simpler and equivalent procedure is to replace $V^{0,l}_{mn}$ and $V^{1,l}_{mn}$ in the scheme
by $g^{-l}e^{im\theta}e^{in\phi}$ for each value of $l,n$ and $m$. Then, the resulting equation can be solved for the amplification factor.\\
Replacing $V^{0,l+1/2}_{mn}$ and $V^{0,l}_{mn}$  by  $\widehat{g}g^{-l}e^{im\theta}e^{in\phi}$ and $g^{-l}e^{im\theta}e^{in\phi}$  respectively, one gets
\begin{equation}\label{8-2}
\begin{split}
&\dfrac{\Delta t}{2}A_{dx}V^{0,l+1/2}_{m,n}=\dfrac{\Delta t}{2}\widehat{g}g^{-l}e^{im\theta}e^{in\phi}(-2\sigma_ 1^2x_m^2\dfrac{sin^2\dfrac{1}{2}\theta }{\Delta x^2}+rx_m\dfrac{isin\theta}{\Delta x})=\widehat{g}g^{-l}e^{im\theta}e^{in\phi}(-a_1sin^2\dfrac{1}{2}\theta +b_1 isin\theta)\\
&\dfrac{\Delta t}{2}A_{dy}V^{0,-l}_{m,n}=\dfrac{\Delta t}{2}g^{-l}e^{im\theta}e^{in\phi}(-2\sigma_ 2^2y_n^2\dfrac{sin^2\dfrac{1}{2}\phi }{\Delta y^2}+ry_n\dfrac{isin\phi}{\Delta y}-r)=g^{-l}e^{im\theta}e^{in\phi}(-a_2sin^2\dfrac{1}{2}\phi +b_2 isin\phi -c_1)\\
&\dfrac{\Delta t}{2}A_{dxdy}V^{0,l+1/2}_{m,n}=\dfrac{\Delta t}{2}\widehat{g}g^{-l}e^{im\theta}e^{in\phi}\sigma_1 \sigma_2 \rho x_my_n(-\dfrac{sin\theta sin\phi}{\Delta x \Delta y})=-\widehat{g}g^{-l}e^{im\theta}e^{in\phi}c_2sin\theta sin\phi \\
&\dfrac{\Delta t}{2}A_{dxdy}V^{0,l}_{m,n}=\dfrac{\Delta t}{2}g^{-l}e^{im\theta}e^{in\phi}\sigma_1 \sigma_2 \rho x_my_n(-\dfrac{sin\theta sin\phi}{\Delta x \Delta y})=-g^{-l}e^{im\theta}e^{in\phi}c_2sin\theta sin\phi.
\end{split}
\end{equation}
Here
\begin{equation}\label{8-3}
\begin{split}
&a_1(x_m)=\frac{\Delta t\sigma _1^2x_m^2}{\Delta x^2},~~~b_1(x_m)=\frac{\Delta t rx_m}{2\Delta x},~~~c_1=\dfrac{r\Delta t}{2},\\
&a_2(y_n)=\frac{\Delta t\sigma _2^2y_n^2}{\Delta y^2},~~~~b_2(y_n)=\frac{\Delta t ry_n}{2\Delta y} ,~~~~c_2(x_m,y_n)=\frac{\Delta t\sigma _1\sigma _2\rho x_my_n}{2\Delta x\Delta y} .
\end{split}
\end{equation}
Also, by replacing $V^{1,l+1/2}_{mn}$ and $V^{1,l}_{mn}$  by  $\widehat{g}g^{-l}e^{im\theta}e^{in\phi}$ and $g^{-l}e^{im\theta}e^{in\phi}$  respectively, we have that
\begin{equation}\label{8-4}
\begin{split}
&\dfrac{\Delta t}{2}A_{dx}V^{1,l+1/2}_{m,n}=\dfrac{\Delta t}{2}\widehat{g}g^{-l}e^{im\theta}e^{in\phi}(-2\sigma_ 1^2x_m^2\dfrac{sin^2\dfrac{1}{2}\theta }{\Delta x^2}+rx_m\dfrac{isin\theta}{\Delta x})=\widehat{g}g^{-l}e^{im\theta}e^{in\phi}(-a_1sin^2\dfrac{1}{2}\theta +b_1 isin\theta)\\
&\dfrac{\Delta t}{2}A_{dy}V^{1,l}_{m,n}=\dfrac{\Delta t}{2}g^{-l}e^{im\theta}e^{in\phi}(-2\sigma_ 2^2y_n^2\dfrac{sin^2\dfrac{1}{2}\phi }{\Delta y^2}+ry_n\dfrac{isin\phi}{\Delta y}-r)=g^{-l}e^{im\theta}e^{in\phi}(-a_2sin^2\dfrac{1}{2}\phi +b_2 isin\phi -c_1)\\
&\dfrac{\Delta t}{2}A_{dxdy}V^{1,l+1/2}_{m,n}=\dfrac{\Delta t}{2}\widehat{g}g^{-l}e^{im\theta}e^{in\phi}\sigma_1 \sigma_2 \rho x_my_n(-\dfrac{sin\theta sin\phi}{\Delta x \Delta y})=-\widehat{g}g^{-l}e^{im\theta}e^{in\phi}c_2sin\theta sin\phi \\
&\dfrac{\Delta t}{2}A_{dxdy}V^{1,l+1}_{m,n}=\dfrac{\Delta t}{2}g^{-l}e^{im\theta}e^{in\phi}\sigma_1 \sigma_2 \rho x_my_n(-\dfrac{sin\theta sin\phi}{\Delta x \Delta y})=-g^{-l}e^{im\theta}e^{in\phi}c_2sin\theta sin\phi .
\end{split}
\end{equation}
According to "Duhamel's principle" we ignore the $G^{l+1}$ and $G^{l}$ terms in stability analysis (more details \cite{Strikwerda}). We obtain the amplification factor 
\begin{equation}\label{8-5}
g=\dfrac{1-a_2sin^2\frac{1}{2}\phi+ b_2 isin\phi -c_1-c_2sin\theta sin\phi}{(1+a_1sin^2\frac{1}{2}\theta -b_1isin\theta)\widehat{g}}
\end{equation}
where
\begin{equation}\label{8-6}
\widehat{g}=\frac{1+a_2sin^2\frac{1}{2}\phi -b_2isin\phi+c_1}{1-a_1sin^2\frac{1}{2}\theta +b_1isin\theta -c_2sin\theta sin\phi}.
\end{equation}
By arranging, one gets
\begin{equation}\label{8-7}
g=\dfrac{[1-a_1sin^2\frac{1}{2}\theta -c_2sin\theta sin\phi +(b_1sin\theta)i][1-a_2sin^2\frac{1}{2}\phi -c_1-c_2sin\theta sin\phi +(b_2 sin\phi)i]}{[1+a_1sin^2\frac{1}{2}\theta -(b_1sin\theta)i][1+a_2sin^2\frac{1}{2}\phi +c_1 -(b_2sin\phi)i]}.
\end{equation}
Thus
 \begin{equation}\label{8-8}
|g(\theta , \phi)|^2=\dfrac{[(1-a_1sin^2\frac{1}{2}\theta -c_2sin\theta sin\phi )^2+b_1^2sin^2\theta][(1-a_2sin^2\frac{1}{2}\phi -c_1-c_2sin\theta sin\phi)^2 +b_2^2 sin^2\phi]}{[(1+a_1sin^2\frac{1}{2}\theta)^2 +b_1^2sin^2\theta][(1+a_2sin^2\frac{1}{2}\phi +c_1)^2 +b_2^2sin^2\phi]}
\end{equation}
According to (\ref{8-3}), we can write $a_2=Ca_1$, $c_2=\widehat{C} a_1$where $C$ and $\widehat{C}$ are constants. Moreover $\dfrac{b_1}{a_1}\rightarrow 0, as ~\Delta x \rightarrow 0,$ so $b_1=\xi a_1$ as $\xi \rightarrow 0$, and $b_2=\xi a_1$, $c_1=\xi a_1$ . Therefore by replacing in (\ref{8-8}), one gets
\begin{equation}\label{8-9}
\begin{split}
\lim_{\xi \rightarrow 0}g^2_{\xi}&=\lim_{\xi \rightarrow 0}\dfrac{[(1-a_1sin^2\frac{1}{2}\theta -c_2sin\theta sin\phi )^2+b_1^2sin^2\theta][(1-a_2sin^2\frac{1}{2}\phi -c_1-c_2sin\theta sin\phi)^2 +b_2^2 sin^2\phi]}{[(1+a_1sin^2\frac{1}{2}\theta)^2 +b_1^2sin^2\theta][(1+a_2sin^2\frac{1}{2}\phi +c_1)^2 +b_2^2sin^2\phi]}\\
&=\dfrac{(1-a_1sin^2\frac{1}{2}\theta -\widehat{C} a_1 sin\theta sin\phi )^2(1-Ca_1sin^2\frac{1}{2}\phi -\widehat{C} a_1sin\theta sin\phi)^2}{(1+a_1sin^2\frac{1}{2}\theta)^2(1+Ca_1sin^2\frac{1}{2}\phi )^2 }.
\end{split}
\end{equation}
 It is enough to find conditions so that
 \begin{equation}\label{8-10}
 \begin{split}
 \dfrac{(1-a_1sin^2\frac{1}{2}\theta -\widehat{C} a_1 sin\theta sin\phi )^2(1-Ca_1sin^2\frac{1}{2}\phi -\widehat{C} a_1sin\theta sin\phi)^2}{(1+a_1sin^2\frac{1}{2}\theta)^2(1+Ca_1sin^2\frac{1}{2}\phi )^2 } \leq 1.
\end{split}
\end{equation}
Notice that
\begin{equation}\label{8-11}
\begin{split}
a_1sin^2\frac{1}{2}\theta +\widehat{C} a_1 sin\theta sin\phi  &\leq a_1\vert sin^2\frac{1}{2}\theta \vert +\widehat{C}a_1\vert  sin\theta sin\phi \vert \\
&\leq a_1\vert sin\frac{1}{2}\theta \vert [\vert sin\frac{1}{2}\theta \vert +2\widehat{C}\vert cos\frac{1}{2}\theta sin\phi\vert]\\
&\leq a_1[1+2\widehat{C}].
\end{split}
\end{equation}
Thus $1-a_1sin^2\frac{1}{2}\theta -\widehat{C} a_1 sin\theta sin\phi \geq 0$, provided that $a_1[1+2\widehat{C}]\leq 1$. and also we have
\begin{equation}\label{8-12}
\begin{split}
Ca_1sin^2\frac{1}{2}\phi +\widehat{C} a_1sin\theta sin\phi &\leq Ca_1\vert sin^2\frac{1}{2}\phi \vert+\widehat{C} a_1\vert sin\theta sin\phi \vert \\
&\leq a_1 \vert sin\frac{1}{2}\phi \vert [ Csin\frac{1}{2}\phi +4\widehat{C} \vert cos\frac{1}{2}\phi sin\theta \vert \\
&\leq a_1[C+2\widehat{C}].
\end{split}
\end{equation}
Then $1-Ca_1sin^2\frac{1}{2}\phi -\widehat{C} a_1sin\theta sin\phi \geq 0$, provided that $a_1[C+2\widehat{C}]\leq 1$. Therefore we should find conditions so that 
\begin{equation}\label{8-13}
 \begin{split}
 \dfrac{(1-a_1sin^2\frac{1}{2}\theta -\widehat{C} a_1 sin\theta sin\phi )(1-Ca_1sin^2\frac{1}{2}\phi -\widehat{C} a_1sin\theta sin\phi)}{(1+a_1sin^2\frac{1}{2}\theta)(1+Ca_1sin^2\frac{1}{2}\phi ) } \leq 1,
\end{split}
\end{equation}
or equivalently
 \begin{equation}\label{8-14}
\begin{split}
a_1(sin^2\frac{1}{2}\theta +\widehat{C} sin\theta sin\phi +Csin^2\frac{1}{2}\phi)(-2+a_1\widehat{C}sin\theta sin\phi)\leq 0.
\end{split}
\end{equation}
Since $\vert y \vert \leq 1$, then for any $x \in R$, $xy\geq -\vert x \vert$, and  by $C\geq 4\widehat{C}^2$, we have that
\begin{equation}\label{8-15}
\begin{split}
sin^2\frac{1}{2}\theta +\widehat{C} sin\theta sin\phi +Csin^2\frac{1}{2}\phi &\geq \vert sin\frac{1}{2}\theta \vert ^2 -4\widehat{C} \vert sin\frac{1}{2}\theta sin\frac{1}{2}\phi \vert  +4\widehat{C}^2 \vert sin\frac{1}{2}\phi \vert ^2\\ &=(\vert sin\frac{1}{2}\theta \vert -2\widehat{C}\vert sin\frac{1}{2}\phi \vert)^2\geq 0.
\end{split}
\end{equation}
Thus (\ref{8-13}) is satisfied if $a_1 \leq \dfrac{2}{\widehat{C}}$ and $ \vert g(\theta, \phi) \vert \leq 1$ if
\begin{equation}\label{8-16}
\begin{split}
a_1 \leq A=min \lbrace \dfrac{2}{\widehat{C}},\dfrac{1}{1+2\widehat{C}},\dfrac{1}{4\widehat{C}^2+2\widehat{C}}\rbrace ~~or~~~\dfrac{\Delta t}{\Delta x^2} \leq \dfrac{A}{\sigma_1^2.x_{max}^2},~~\dfrac{\Delta t}{\Delta y^2} \leq \dfrac{A}{\sigma_2^2.y_{max}^2}.
\end{split}
\end{equation}
Since $\Delta x=\Delta y$ and $x_{max}=y_{max},$ a sufficient condition for stability of the scheme is
\begin{equation}\label{8-17}
\begin{split}
\dfrac{\Delta t}{\Delta x^2} \leq \dfrac{A}{max \lbrace \sigma_1^2, \sigma_2^2 \rbrace x_{max}^2}.
\end{split}
\end{equation}
Thus, the Peaceman-Rachford scheme is stable if the number of steps in the time interval, $L$, and in the spatial domain, $M=N$, satisfy inequality (\ref{8-17}). This condition is a consequence of the cross-derivative terms. In the absence of these terms, the scheme would be unconditionally stable.

The remaining issue we need to address is the convergence of the numerical method to the true value of the problem. According to \cite{Strikwerda} , this scheme is first-order accurate in time and space and due to stability the scheme is convergent. Results of this convergence are summarized in the next section.

\section{Numerical Results}\label{section5}
 Let us fix the values of the parameters of the marginal dynamical equations according to Table\ref{table1}. We also assume the following form for price impact
 \begin{align*}
\lambda(t) =
\left\{
	\begin{array}{ll}
	\varepsilon(1-e^{-\beta(T-t)^{3/2}}),  &  \underline{S} \leqslant S_1 \leqslant \overline{S}, \\
		0, & otherwise,
	\end{array}
\right.
\end{align*}
where $\varepsilon$ is a constant price impact coefficient, $T-t$ is time to expiry, $\beta$ is a decay coefficient, $\underline{S}$ and $\overline{S}$ represent respectively, the lower and upper limit of the stock price within which there is a impact price.\\
We consider $\underline{S}=60, \overline{S}=140,\varepsilon=0.01$  and $\beta=100$ for the subsequent numerical analysis. Choosing a different value for $\beta, \underline{S}$ and $\overline{S}$ will change the magnitude of the subsequent results, however, the main qualitative results remain valid.

\begin{table}[h]
\begin{center}
\begin{tabular}{|c|c|c|c|c|}
\hline
  & $S(t_0)$ & $\sigma$  & $S_{min}$ &$S_{max}$  \\ 
\hline Asset 1    & 112 &0.15& 0&200  \\ 
\hline  Asset 2  &  104 & 0.10& 0& 200\\  
\hline
 \end{tabular}
\end{center} 
\caption{\footnotesize{Model data together with $r = 0.04$}}\label{table1} 
\end{table}
{\bf Convergence of numerical results.}   As we mentioned in Section\ref{section1}, the exact option values for the option in illiquid market are unknown. Since $\lambda=0$ leads to the standard Black-Scholes model, we compare the results obtained from the numerical method (with $\lambda =0$ and strike $0$) with the Margrabe's closed formula for exchange options (i.e. Spread Option with strike $0$). We fix the values of the parameters 
\begin{figure}
\begin{center}
   \includegraphics[width=13cm]{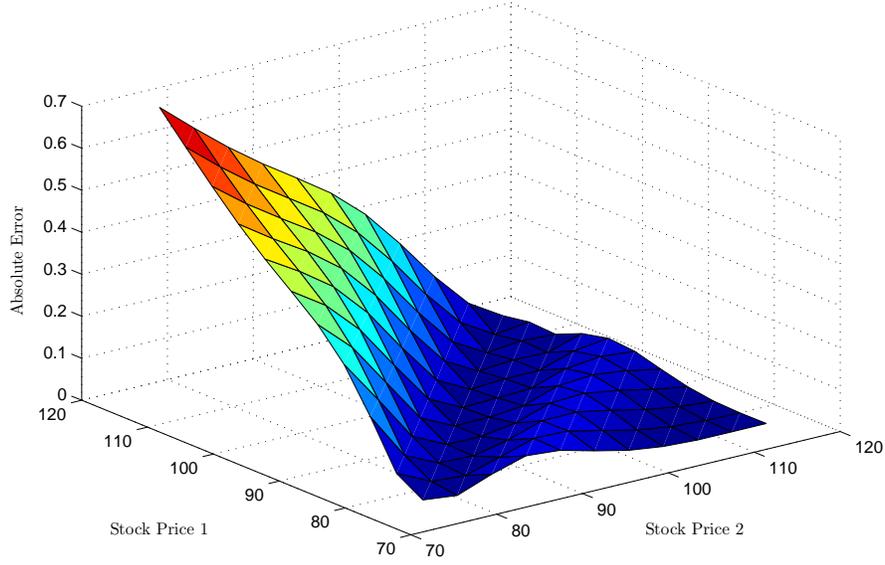}\\
    \caption{\footnotesize {Absolute errors between our approximation  and Margrabe's closed formula, with $\sigma_1=0.15, \sigma_2=0.10, r=0.05,\rho=0.7, T=0.7$ year, $m=50$ and $l=100$. }  }\label{fig1}
    \end{center}
\end{figure}
\begin{table}[!hbt]
\begin{center}
 \begin{tabular}{|c|c|c|c|c|c|c|c|c|}
 \hline
\multicolumn {1}{|c|}{}&{m}&{l} &{$T=0.1$}    & {$T=0.3$}  &{$T=0.5$}    &{$T=0.7$}         &{$ T=1 $ }   \\ 
 \hline $\rho=0.1$  & 50   & 100    &8.1979   &9.1570   &10.0519    &10.8369  &11.8622 \\
                             & 100 & 100    &8.2110   &9.1892   &10.0930    &10.8757  &11.9579    \\
                             & 200 & 200    &8.2153   &9.2373   &10.1607    &10.9727  &12.0041     \\ 
 Margrabe     &   &  &  \textit{8.2323}  & \textit{9.2462}  & \textit{10.1723}     &\textit{10.9892}    &\textit{12.0666}  \\ 
 \hline $\rho=0.5$  & 50   &100     &8.0088     &8.5425   &9.1276     &9.6662      &10.5095  \\                                                          
                             & 100 &100     &8.0591     &8.5983   &9.1961     &9.7205      &10.5405  \\
                             & 200 &200     &8.0687     &8.6222   &9.2209     &9.7843      &10.5636   \\
 Margrabe           &        &         & \textit{8.0692}         & \textit{8.6235}         &\textit{9.2294}      &\textit{9.7949}       &\textit{10.5648}   \\
 
\hline $\rho=0.7$   & 50   &100    &7.9195        &8.2199  &8.6180    & 9.0019  &9.5315   \\                                                           
                             & 100 &210    &7.9734        &8.2509  &8.6296    &9.0929   &9.6244   \\
                             &200  &200    &7.9950        &8.3023  &8.7106    &9.1035   &9.6728     \\
 Margrabe          &        &          &  \textit{8.0186}        & \textit{8.3128}        & \textit{8.7115}      &\textit{9.1110}       &\textit{9.6775}    \\
\hline $\rho=0.9$  & 50    &100    &7.9252   &7.9803  &8.1740    &8.3417   &8.6412     \\                                                            
                            & 100  &100    &7.9310   &7.9852  &8.1894  &8.3532   &8.6498     \\ 
                            & 200  &200    &7.9938   &8.0515  & 8.2032   &8.3686   &8.6571     \\
 Margrabe          &        &          &  \textit{8.0005}   & \textit{8.0588}        &\textit{8.2015}       &\textit{8.3799}       &\textit{8.6675}    \\ 
 \hline
 \end{tabular} 
\end{center}
\caption{\footnotesize{Convergence of the Peaceman-Rachford method to Magrabe formula. Data are given in Table\ref{table1}.}}\label{table2}
\end{table}
according to Table\ref{table1}, and vary the values of the correlation coefficient $\rho$. Results of this convergence study are summarized in Table\ref{table2}. We can see from the table that the agreement is excellent. We plot the absolute error of our approximation (using $\lambda =0,$ strike $0$ and Margrabe's closed formula as benchmark) against the stocks in Fig\ref{fig1}. Results of the numerical method for Spread option in illiquid market are stated in Table\ref{table3}. 
\begin{table}[!hbt]
 \begin{center}
 \begin{tabular}{|c|c|c|c|c|c|c|c|c|c|c|c|}
 \hline
\multicolumn {1}{|c|}{} &{$k=-15$}  & {$k=-5$} &{$k=-2$} & {$k=0$} &{$k=2$}   & {$k=5$}  & {$k=10$}   & {k=20}   \\ 
 \hline $\rho=0.1$    & 15.0929        & 7.1600      & 5.3275      & 4.2936   & 3.4027    & 2.3395      & 1.1267       & 0.1905   \\
   Excess price       &  0.0001         & 0.0005      & 0.0005      & 0.0005   & 0.0005    & 0.0005      & 0.0003       & 0.00006   \\
 \hline $\rho=0.5$   &  14.7992       & 6.2972      & 4.3645      & 3.3368   & 2.4486    & 1.4909      & 0.5435       & 0.0426    \\                                                                                                                                    
   Excess price       &  0.0001         & 0.0007      & 0.0009      & 0.0009   & 0.0009    & 0.0007      & 0.0003       & 0.00003    \\
 \hline $\rho=0.7 $  &  14.7085       & 5.7956      & 3.7731      & 2.7085   & 1.8642    & 0.9981      & 0.2593       & 0.0055   \\                                                                                 
   Excess price         &  0.00006       & 0.0009      & 0.0013      & 0.0013   & 0.0012    & 0.0009      & 0.0004       & 0.00001   \\
 \hline $\rho=0.9$  &  14.6833       & 5.2299      & 3.0523      &1.9601    & 1.1531    & 0.4387      & 0.0088       & 0.0029    \\                                                                       
   Excess price       &  0.00003       & 0.0013      &  0.0020     &0.0020    & 0.0018    & 0.0012      & 0.0003       & 0.00000   \\   \hline
\end{tabular} 
\end{center}
\caption{\footnotesize{ The values of a 0.4 year European call Spread
option based on different correlation, and strikes. Excess price shows the difference in call Spread option from Black-Scholes. The values of the parameters used for these runs
are $\sigma_1=0.15, \sigma_2=0.10, r=0.05$ with $m = l = 100$.}}\label{table3} 
\end{table}

{\bf Replicating cost.}
Next investigate the effects of the price impact (full feedback model) on the replication cost of Spread option. We investigate the excess price which is the difference between the call price in the full feedback model and the corresponding Black-Scholes price. 

\begin{figure}[!hbt]
\begin{center}
   \includegraphics[width=13cm]{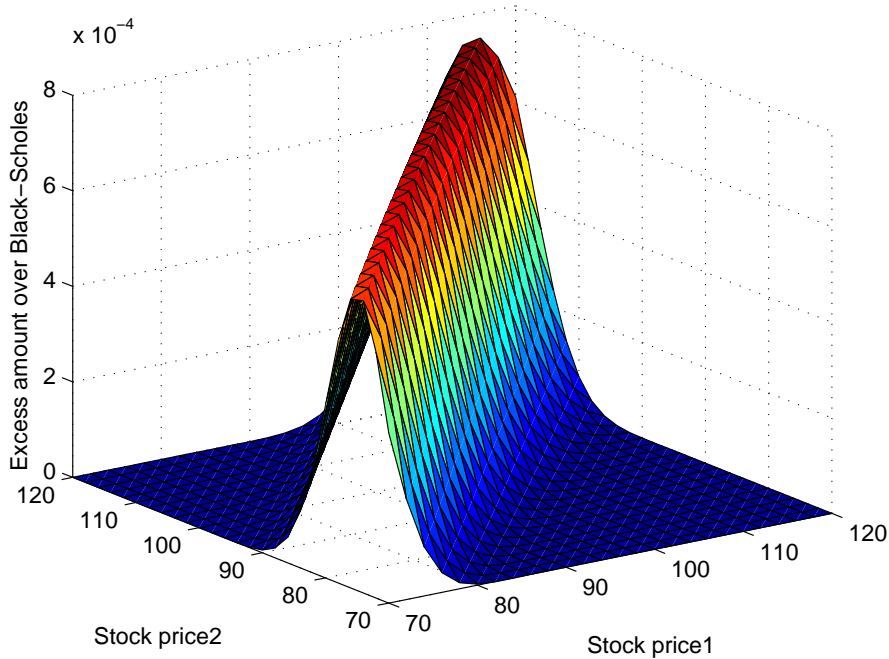}
    \caption{\footnotesize {The call price difference (full feedback model and classical model) as a function of stock price at time $0$ against $S_1$ and $S_2$. $K=5, \sigma_1=0.3, \sigma_2=0.2, r=0.05, \rho=0.7, T=0.1$, and $m=l=100$.}  }\label{fig2}
    \end{center}
\end{figure}

\begin{figure}[!hbt]
\begin{center}
   \includegraphics[width=12cm]{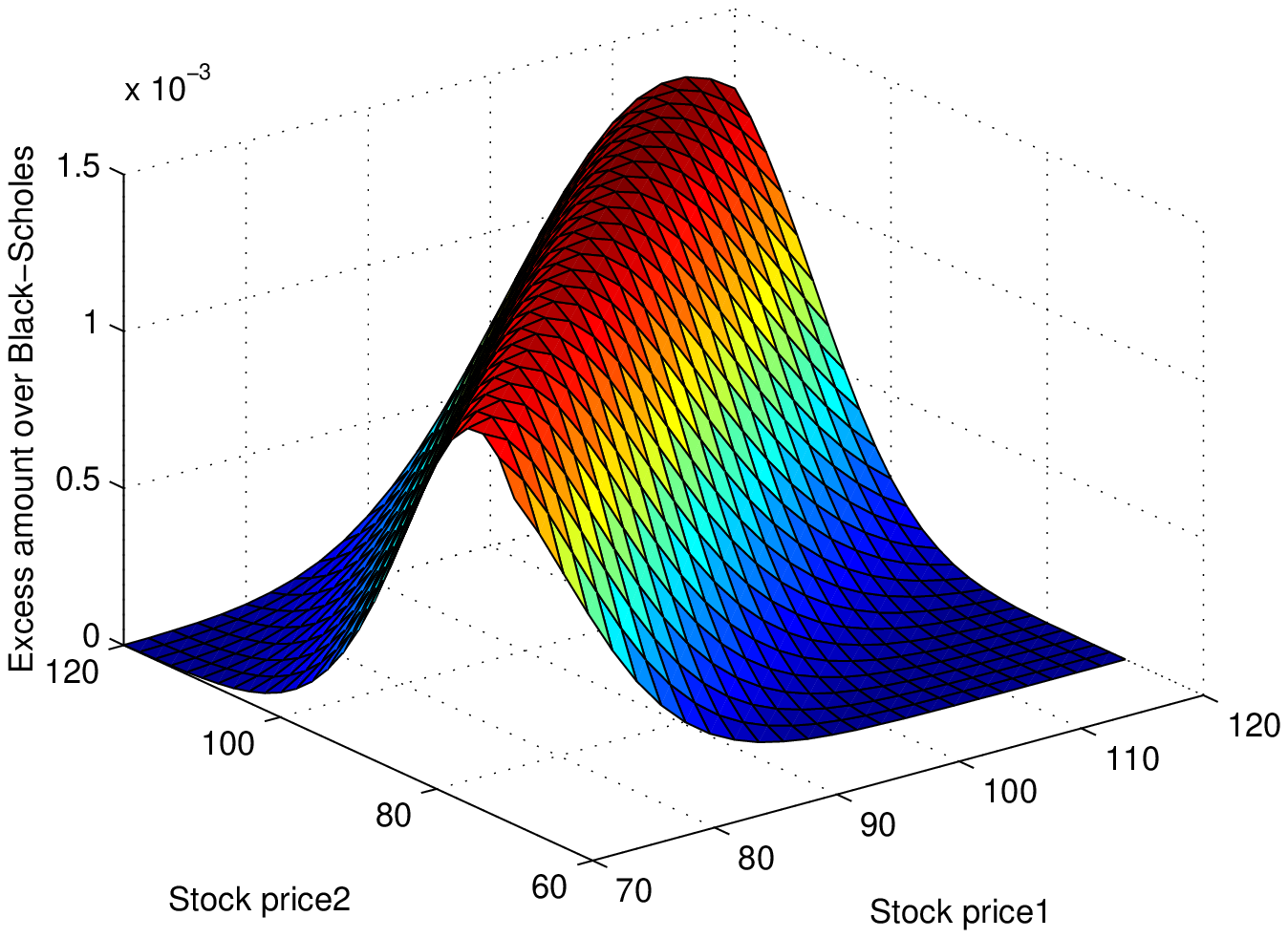}
    \caption{\footnotesize {The call price difference (full feedback model and classical model) as a function of stock price at time $0$ against $S_1$ and $S_2$. $K=5, \sigma_1=0.3, \sigma_2=0.2, r=0.05, \rho=0.7, T=0.4$, and $m=l=100$.
 }  }\label{fig3}
    \end{center}
\end{figure}

\newpage
\begin{figure}[!hbt]
\begin{center}
   \includegraphics[width=12cm]{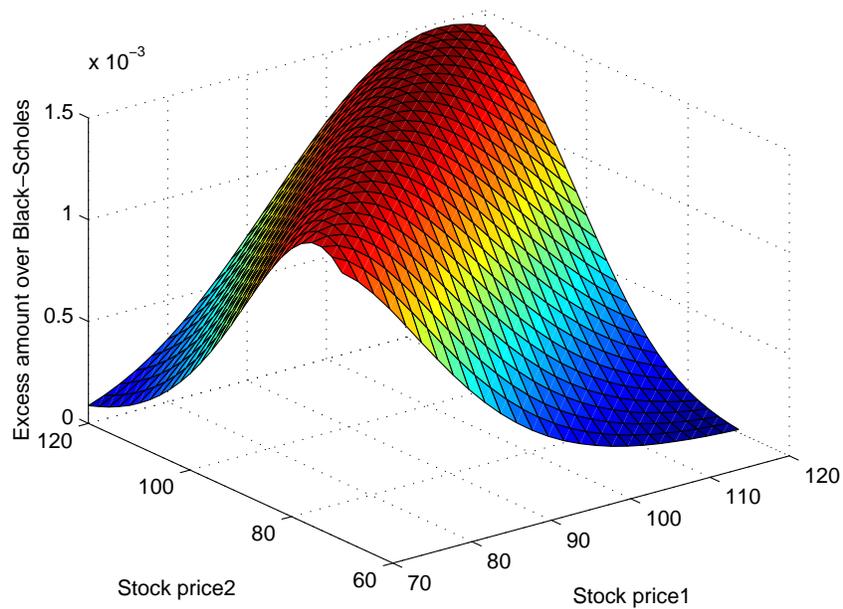}
    \caption{\footnotesize {The call price difference (full feedback model and classical model) as a function of stock price at time $0$ against $S_1$ and $S_2$. $K=5, \sigma_1=0.3, \sigma_2=0.2, r=0.05, \rho=0.7, T=1$, and $m=l=100$.} }\label{fig4}
    \end{center}
\end{figure}
These figures indicate that the Spread option price in the full feedback model is higher than the classical Spread option price.

\begin{figure}[!hbt]
\begin{center}
   \includegraphics[width=19cm]{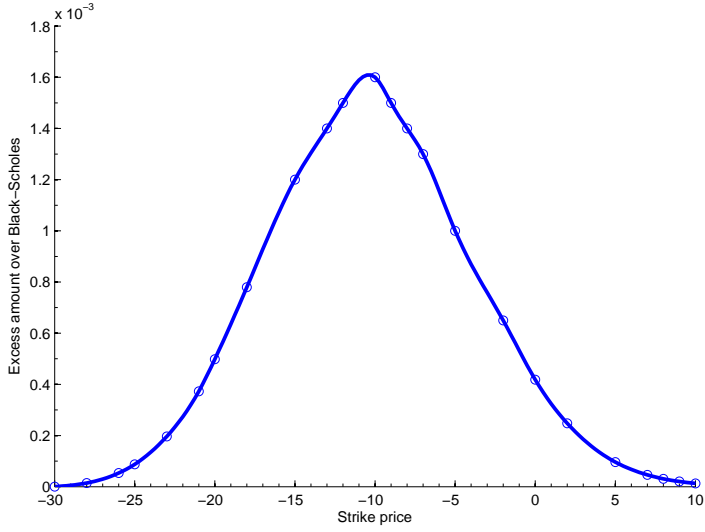}
    \caption{\footnotesize {The call price difference (full feedback model and classical model) against the strike price $K$. $S_1(t_0)=100, S_2(t_0)=110, \sigma_1=0.15, \sigma_2=0.10, r=0.05, \rho=0.7$, $T=0.4$ year and $m=l=100$.  
}  }\label{fig5}
    \end{center}
\end{figure}
As the option becomes more and
more in the money and out of the money, the excess price converges monotonically to
zero.  

\section{Conclusion}\label{section6}
In this work, we have investigated a model which incorporates illiquidity of the underlying asset into the
classical multi-asset Black-Scholes-Merton framework. We considered the full feedback model in which the hedger is assumed to be aware of the feedback effect and so would change the hedging strategy accordingly. Since there is no analytical formula for the price of an option within this model, we applied the Matched Asymptotic Expansions technique to linearize the partial differential equation characterizing the price. 
 We applied a standard alternating direction implicit method (Peaceman-Rachford scheme) to solve the corresponding linear equations numerically. We also discussed the stability and the convergence of the numerical scheme. By running a numerical experiment, we investigated the effects of liquidity on the Spread option pricing in the full feedback model. Finally, we found out that the Spread option price in the market with finite liquidity (full feedback model), is more than the Spread option price in the classical Black-Scholes-Merton framework.

\end{document}